\documentclass[11pt,leqno]{amsart}

\usepackage{amsmath}
\usepackage{amssymb}
\usepackage{a4wide}
\usepackage{graphics}
\usepackage{epsfig}
\usepackage[sans]{dsfont}

\newtheorem{theorem}{Theorem}[section]
\newtheorem{proposition}[theorem]{Proposition}
\newtheorem{lemma}[theorem]{Lemma}
\newtheorem{corollary}[theorem]{Corollary}
\newtheorem{remark}[theorem]{Remark}
\newtheorem{example}[theorem]{Example}
\newtheorem{definition}[theorem]{Definition}

\def\C{{\mathbb C}}

\def\N{{\mathbb N}}

\def\R{{\mathbb R}}
\def\S{{\mathbb S}}

\def\Z{{\mathbb Z}}

\def\CA{{\mathcal A}}
\def\CC{{\mathcal C}}
\def\CD{{\mathcal D}}

\def\CH{{\mathcal H}}

\def\CO{{\mathcal O}}

\def\CS{{\mathcal S}}

\newcommand{\im}{\operatorname{Im}}
\newcommand{\re}{\operatorname{Re}}
\newcommand{\supp}{\operatorname{supp}}
\newcommand{\op}{\operatorname{Op}}
\newcommand{\ad}{\operatorname{ad}}
\newcommand{\rank}{\operatorname{Rank}}

\def\one{\mathds{1}}
\def\<{\langle}
\def\>{\rangle}

\makeatletter
 \@addtoreset{equation}{section}
 \makeatother

\title[Semiclassical estimates of the cut-off resolvent]
{Semiclassical estimates of the cut-off resolvent for trapping perturbations}

\author[J. F. Bony]{Jean-Fran\c{c}ois Bony}
\address{Institut de math\'ematiques, Universit\'e Bordeaux I, 351 cours de la Lib\'eration, 33405 Talence, France}
\email{bony@math.u-bordeaux1.fr}
\author[V. Petkov]{Vesselin Petkov}
\address{Institut de math\'ematiques, Universit\'e Bordeaux I, 351 cours de la Lib\'eration, 33405 Talence, France}
\email{petkov@math.u-bordeaux1.fr}



\begin{document}

\begin{abstract}
This paper is devoted to the study of the cut-off resolvent of a semiclassical ``black box'' operator $P$. We estimate the norm of $\varphi ( P - z )^{- 1} \varphi$, for any $\varphi \in C^{\infty}_{0} ( \R^{n} )$, by the norm of $\one_{\CC_{a,b}} ( P - z )^{- 1} \one_{\CC_{a,b}}$ where $\CC_{a,b} = \{ x \in \R^{n} ; \ a < \vert x \vert < b \}$ and $a \gg 1$. For $z$ in the unphysical sheet with $- M h \vert \ln h \vert \leq \im z \leq 0$, we prove that this estimate holds with a constant $\frac{h}{\vert \im z \vert} e^{C \vert \im z \vert / h}$. We also study the resonant states $u$ of the operator $P$ and we obtain bounds for $\Vert \varphi u \Vert$ by $\Vert \one_{\CC_{a,b}} u \Vert$. These results hold without any assumption on the trapped set nor any assumption on the multiplicity of the resonances.
\end{abstract}

\maketitle

\section{Introduction} \label{s1}

In this paper, we prove estimates on the meromorphic extension across the real axis of the cut-off resolvent of $P$, a semiclassical operator of ``black box'' type. This abstract framework, introduced by Sj\"{o}strand and Zworski \cite{SjZw91_01} and described below, allows one to develop the theory of resonances for many kinds of perturbations (potentials, obstacles, metrics, \dots). In particular, the results stated below hold for arbitrary dimension $n \geq 1$ and without any restriction on the geometry of the trapped set.

More precisely, we will estimate the norm of the cut-off resolvent $\varphi ( P - z )^{- 1} \varphi$, for any $\varphi \in C^{\infty}_{0} ( \R^{n} )$, by the norm of $\one_{\CC_{a,b}} ( P - z )^{- 1} \one_{\CC_{a,b}}$ where
\begin{equation*}
\CC_{a,b} = \{ x \in \R^{n} ; \ a < \vert x \vert < b \} .
\end{equation*}
Notice that, on the real axis, there is a big contrast between the behavior of these two norms. Indeed, the resolvent truncated on rings $\CC_{a,b}$, with $1 \ll a < b$, is always bounded above by $C h^{- 1}$. On the other hand, the norm of the resolvent, truncated near the projection on $\R^{n}$ of the trapped set, depends on the geometry of this set and can be much larger than $h^{- 1}$. For scattering outside a bounded obstacle $K \subset \R^{n}$, with $n \geq 3$ odd, a similar question has been investigated by Stoyanov and the second author \cite{PeSt09_01}. Using the scattering theory of Lax and Phillips \cite{LaPh89_01}, they have proved that the cut-off resolvent can be bounded by the norm of the scattering matrix (we refer to Section \ref{s6} for more details).

In scattering theory, it is natural to consider the resolvent of $P$ truncated in rings $\CC_{a,b}$ far away from the origin. Indeed, the operator $\one_{\CC_{a,b}} ( P - z )^{- 1} \one_{\CC_{a,b}}$ appears in the representation of the scattering amplitude for compact perturbations. More precisely, assume that $P$ is a compactly supported perturbation of $- h^{2} \Delta$ and denote by $S ( z ; h ) = I + K ( z ; h)$ the associated scattering matrix at energy $z$. By definition, the scattering amplitude $a ( z , \omega , \omega^{\prime} ; h )$ is the distribution kernel of $K ( z ; h )$. The standard formula (see for instance, Zworski and the second author \cite{PeZw01_01}) gives
\begin{equation} \label{a52}
a ( z , \omega , \omega^{\prime} ; h ) = c ( z ; h )  \Big\<  e^{i \sqrt{z} \< x, \omega \> / h} , [ h^{2} \Delta , \chi_{1} ] ( P - z )^{- 1} [ h^{2} \Delta , \chi_{2} ] e^{i \sqrt{z} \< x , \omega^{\prime} \> / h} \Big\> ,
\end{equation}
where $\chi_{1} , \chi_{2} \in C_{0}^{\infty} ( \R^{n} )$ are cut-off functions, $\omega , \omega^{\prime} \in \S^{n - 1}$ and 
\begin{equation*}
c ( z , h ) = i \pi ( 2 \pi h )^{- n} z^{\frac{n - 2}{2}} .
\end{equation*}
Moreover, we can take the functions $\chi_{1} , \chi_{2}$ equal to $1$ on arbitrary large compact sets containing the perturbation, and the scattering amplitude is independent of this choice. Thus the estimatation of $\one_{\CC_{a,b}} ( P - z )^{- 1} \one_{\CC_{a,b}}$ with $1 \ll a < b$ is essential for the estimations of the scattering amplitude and for the norm of the Hilbert--Schmidt operator $K ( z ; h )$.

We now give the precise assumptions on the semiclassical ``black box'' operator $P$. This was introduced by Sj\"{o}strand and Zworski \cite{SjZw91_01} (see also Sj\"{o}strand \cite{Sj97_01,Sj01_01,Sj07_01} in the long range case). Let $\CH$ be a complex Hilbert space with an orthogonal decomposition
\begin{equation*}
\CH = \CH_{R_0} \oplus L^2 ( \R^n \setminus B( R_0 ) ) ,
\end{equation*}
with $n \geq 1$, $R_0 > 0$ and $B ( R ) = \{ x \in \R^n ; \ \vert x \vert < R \}$. In the sequel, we will identify $u \in L^2 ( \R^n \setminus B( R_0 ) )$ with $0 \oplus u \in \CH$. We consider a self-adjoint semiclassical operator $P : \CH \longrightarrow \CH$ with domain $\CD$ independent of $h \in ] 0 , 1 ]$. We assume that
\begin{equation*}
\one_{\R^{n} \setminus B ( R_{0} )} \CD = H^{2} ( \R^{n} \setminus B ( R_{0} ) ) ,
\end{equation*}
and conversely that any $u \in H^{2} ( \R^{n} \setminus B ( R_{0}))$, which vanishes near $\partial B ( R_{0} )$, is an element of $\CD$. To treat the contribution of $P$ in $\CH_{R_0}$, we suppose that
\begin{equation*}
\one_{B ( R_{0} )} ( P + i )^{- 1} \text{ is compact.}
\end{equation*}

We also assume that, for all $u \in \CD$, we have
\begin{equation*}
\one_{\R^{n} \setminus B ( R_{0} )} P u = Q ( u \vert_{\R^{n} \setminus B( R_{0} )} ) ,
\end{equation*}
where $Q$ is a self-adjoint semiclassical differential operator on $L^{2} ( \R^{n} )$
\begin{equation} \label{a3}
Q = \sum_{\vert \alpha \vert \leq 2} a_{\alpha} ( x ; h ) ( h D_{x} )^{\alpha} .
\end{equation}
We suppose that the $a_{\alpha}$'s are bounded in $C_{b}^{\infty} ( \R^{n} )$ (the space of smooth functions which are bounded with all their derivatives) when $h$ varies, and that $a_{\alpha} ( x ; h ) = a_{\alpha} ( x )$ is independent of $h$ for $\vert \alpha \vert = 2$. We further assume that $Q$ is elliptic:
\begin{equation}
\sum_{\vert \alpha \vert = 2} a_{\alpha} ( x ) \xi^{\alpha} \gtrsim \xi^{2} ,
\end{equation}
and a long range perturbation of the Laplacian:
\begin{equation} \label{d1}
\sum_{\vert \alpha \vert \leq 2} a_{\alpha} ( x ; h ) \xi^{\alpha} \longrightarrow \xi^{2} ,
\end{equation}
as $\vert x \vert \to + \infty$ uniformly with respect to $h$. Finally, we assume that
\begin{equation} \label{a50}
a_{\alpha} ( x ; h ) = a_{\alpha}^{0} ( x ) + h a_{\alpha}^{1} ( x ; h ) ,
\end{equation}
where $a_{\alpha}^{0} , a_{\alpha}^{1} \in C_{b}^{\infty} ( \R^{n} )$ uniformly with respect to $h$. We denote by
\begin{equation} \label{d2}
q ( x , \xi ) = \sum_{\vert \alpha \vert \leq 2} a_{\alpha}^{0} ( x ) \xi^{\alpha} ,
\end{equation}
the semiclassical principal symbol of $Q$.

To define the resonances, we assume that the coefficients $a_{\alpha} ( x ; h )$ extend holomorphically in $x$ to the region
\begin{equation}
\Upsilon = \big\{ x \in \C^{n} ; \ \vert \im x \vert \leq \delta \vert \re x \vert \text{ and } \vert \re x \vert \geq R_{1} \big\} ,
\end{equation}
for some $\delta > 0$ and $R_{1} > R_{0}$, and that the relevant parts of \eqref{a3}--\eqref{a50} remain valid in $\Upsilon$. Under these assumptions, it is possible to define the resonances by complex distortion following the approach of Sj\"{o}strand \cite{Sj01_01} (see also Aguilar and Combes \cite{AgCo71_01}, Hunziker \cite{Hu86_01}, Hellfer and Martinez \cite{HeMa87_01} and Sj\"{o}strand and Zworski \cite{SjZw91_01} for more references concerning the definition of the resonances by complex scaling). Let $\Gamma_{\theta}$ be a maximally totally real manifold which coincides with $\R^{n}$ along $B ( R_{1} )$ and with $e^{i \theta} \R^{n}$ outside a compact set, and which satisfies some additional assumptions described in \cite[Section 3]{Sj01_01}. For $0 \leq \theta \leq \theta_{0}$ with $\theta_{0} > 0$ small enough, the operator
\begin{equation*}
P_{\theta} = P \vert_{\Gamma_{\theta}} ,
\end{equation*}
is well defined on $\CD$. Moreover, the spectrum of $P_{\theta}$ in
\begin{equation}
\Lambda_{\theta} = \{ z \in \C ; \ - 2 \theta < \arg z \leq 0 \} ,
\end{equation}
is discrete and independent of $\theta$ and of the choice of $\Gamma_{\theta}$ (in the sense that $P_{\theta}$ and $P_{\theta^{\prime}}$ have the same eigenvalues with the same multiplicity in $\Lambda_{\theta} \cap \Lambda_{\theta^{\prime}}$). By definition, the resonances of $P$ are the eigenvalues of $P_{\theta_{0}}$ in $\Lambda_{\theta_{0}}$.

As a matter of fact, the resolvent
\begin{equation*}
( P - z )^{-1} : \CH_{\rm comp} \longrightarrow \CD_{\rm loc} ,
\end{equation*}
admits a meromorphic continuation from the upper complex half-plane $\{ \im z > 0 \}$ to $\Lambda_{\theta_{0}}$ and the poles of this extension are the resonances. Moreover, if a cut-off function $\varphi \in C_0^{\infty}(\R^n)$ is supported in the set where $\Gamma_{\theta}$ coincides with $\R^{n}$, then
\begin{equation} \label{a34}
\varphi ( P - z )^{-1} \varphi = \varphi ( P_{\theta} - z )^{-1} \varphi .
\end{equation}
We refer to Helffer and Martinez \cite{HeMa87_01} for the equivalence of various definitions of the resonances.

For two functions $f , g$, we will use the notation $f \prec g$ if $g = 1$ in a neighborhood of the support of $f$. Since we work with operators of ``black box'' type, the different cut-off functions appearing in the sequel will be assumed to be constant near $B ( R_{0} )$. In the following, $\Vert \cdot \Vert$ will denote the norm of the Hilbert space $\CH$ and the operator norm on $\CH$. Finally, $( P - z )^{- 1}$ will designate the meromorphic extension of the resolvent from the upper half-plane to $\Lambda_{\theta_{0}}$ (and not the inverse of $P - z$). Our first theorem yields a link between the cut-off resolvents with two cut-off functions $\chi$ and an arbitrary cut-off $\varphi$.

\begin{theorem}\sl \label{a13}
Let $[ E_{0} , E_{1} ] \subset ] 0 , + \infty [$. There exists $a_{0} > R_{0}$ such that, for all $M > 0$ and $\chi , \varphi \in C^{\infty}_{0} ( \R^{n} )$ with $\one_{B ( a_{0} )} \prec \chi$, there exists $C > 0$ such that
\begin{equation*}
\big\Vert \varphi ( P - z )^{-1} \varphi \big\Vert \leq C e^{C \vert \im z \vert / h} \big\Vert \chi ( P - z )^{-1} \chi \big\Vert ,
\end{equation*}
for $z \in [ E_{0} , E_{1} ] - i [ 0 , M h \vert \ln h \vert ]$ not a resonance and $h$ small enough.
\end{theorem}

On the real axis, such a result was essentially obtained by Robert and Tamura \cite[Page 437]{RoTa87_01} (see also Bruneau and the second author \cite[Proposition 3]{BrPe00_01} for trapping situations) to prove the well-known resolvent estimate in non-trapping semiclassical situations. The next theorem is our main result. We obtain an estimate of $\varphi (P - z)^{-1} \varphi$ by the norm of the cut-off resolvent $\one_{\CC_{a,b}} (P - z)^{-1} \one_{\CC_{a,b}}$.

\begin{theorem}\sl \label{a1}
Let $[ E_{0} , E_{1} ] \subset ] 0 , + \infty [$. There exists $a_{0} > R_{0}$ such that, for all $a_{0} < a < b$, $M > 0$ and $\varphi \in C^{\infty}_{0} ( \R^{n} )$, there exists $C > 0$ such that
\begin{equation*}
\big\Vert \varphi ( P - z )^{-1} \varphi \big\Vert \leq C \frac{h}{\vert \im z \vert} e^{C \vert \im z \vert / h} \big\Vert \one_{\CC_{a,b}} (P - z)^{-1} \one_{\CC_{a,b}} \big\Vert ,
\end{equation*}
for $z \in [ E_{0} , E_{1} ] - i [ 0, M h \vert \ln h \vert ]$ not a resonance and $h$ small enough.
\end{theorem}

In particular, both Theorem \ref{a13} and Theorem \ref{a1} hold for any $a_{0}$ large enough. The above theorem gives no information on the real axis due to the factor $\vert \im z \vert^{-1}$ in the right hand side. This is in agreement with already known results, which say that the behavior of the resolvent truncated near the trapped set can be very different from its behavior truncated in rings far away from the origin. Indeed, under some additional assumptions on the operator $P$, Burq \cite{Bu02_01} and Cardoso and Vodev \cite{CaVo02_01} have proved that
\begin{equation*}
\sup_{z \in [ E_{0} , E_{1} ]} \big\Vert \one_{\CC_{a,b}} (P - z)^{-1} \one_{\CC_{a,b}} \big\Vert \lesssim h^{-1} ,
\end{equation*}
without hypothesis on the trapped set. On the other hand,
\begin{equation*}
\sup_{z \in [ E_{0} , E_{1} ]} \big\Vert \varphi ( P - z )^{-1} \varphi \big\Vert ,
\end{equation*}
can be of order $h^{-1}$ in the non-trapping case (see Robert and Tamura \cite{RoTa87_01}) or greater than $e^{\varepsilon / h}$, with $\varepsilon > 0$, as in the well in an island situation (see e.g. Helffer and Sj\"{o}strand \cite{HeSj86_01} or Nakamura, Stefanov and Zworski \cite{NaStZw03_01}). For $\im z = - A h$, our result implies the following

\begin{corollary}\sl \label{a12}
Under the assumptions and notations of Theorem \ref{a1} and for $A > 0$, we have
\begin{equation*}
\big\Vert \varphi ( P - z )^{-1} \varphi \big\Vert \lesssim \big\Vert \one_{\CC_{a,b}} (P - z)^{-1} \one_{\CC_{a,b}} \big\Vert ,
\end{equation*}
for $z \in [ E_{0} , E_{1} ] - i A h$ not a resonance.

In particular, if in addition $\varphi$ does not vanish near $B ( a_{0} )$, the norms of the operators $\varphi ( P - z )^{-1} \varphi$ and $\one_{\CC_{a,b}} (P - z)^{-1} \one_{\CC_{a,b}}$ are equivalent for $z \in [ E_{0} , E_{1} ] - i A h$ not a resonance.
\end{corollary}

The term $e^{C \vert \im z \vert / h}$ appearing in Theorem \ref{a13} and Theorem \ref{a1} cannot be removed in general. To show this, it is enough to consider the distribution kernel of $( - h^{2} \Delta - z )^{-1}$ in dimension $n = 1$ which is given by
\begin{equation*}
\frac{i e^{i \sqrt{z} \vert x - y \vert / h}}{2 h \sqrt{z}} .
\end{equation*}
Note also that the constant $C > 0$ in the term $e^{C \vert \im z \vert / h}$ depends necessarily on $a , b , \varphi$.

\begin{remark}\sl
If $P$ has no resonance in $[ E_{0} - \varepsilon , E_{1} + \varepsilon ] - i [ 0 , A h ]$, $\varepsilon > 0$, and if the norm of $\one_{\CC_{a,b}} (P - z)^{-1} \one_{\CC_{a,b}}$ is controlled in $[ E_{0} - \varepsilon , E_{1} + \varepsilon ] - i A h$, one can exploit Corollary \ref{a12} combined with a priori bounds on the cut-off resolvent (see e.g. Burq and Zworski \cite{BuZw01_01}) and the semiclassical maximum principle (see Tang and Zworski \cite{TaZw98_01}) to establish a bound of the cut-off resolvent $\varphi ( P - z )^{-1} \varphi$ without $\vert \im z \vert^{-1}$ in the band $[ E_{0} , E_{1} ] - i [ 0 , A h ]$.
\end{remark}

In the proof of the previous results, we will use the following lower bound which can have an independent interest.

\begin{proposition}\sl \label{a10}
Let $[ E_{0} , E_{1} ] \subset ] 0 , + \infty [$. There exists $a_{0} > R_{0}$ such that, for all $\varphi \in C^{\infty}_{0} ( \R^{n} )$ satisfying $\supp \varphi \cap B ( a_{0} )^{c} \neq \emptyset$, there exists $C > 0$ such that
\begin{equation*}
\big\Vert \varphi ( P- z )^{- 1} \varphi \big\Vert \geq C h^{-1} e^{- \vert \im z \vert / h} ,
\end{equation*}
for $z \in ( [ E_{0} , E_{1} ] - i [ 0 , 1 ] ) \cap \Lambda_{\theta_{0} / 2}$ not a resonance and $h$ small enough.
\end{proposition}

The second question we deal with in this paper is that of estimating resonant states. Let $z$ be a resonance of $P$. Then, from the general theory of resonances, we can write, for $\lambda$ in a neighborhood of $z$,
\begin{equation} \label{a51}
( P - \lambda )^{-1} = \frac{\Pi_{N}}{(z - \lambda )^{N}} + \cdots + \frac{\Pi_{1}}{z - \lambda} + \CA ( \lambda ) ,
\end{equation}
as operators from $\CH_{\rm comp}$ to $\CD_{\rm loc}$, where $\CA ( \lambda )$ is an operator-valued function holomorphic near $z$ and the $\Pi_{j}$'s are finite rank operators satisfying $\im \Pi_{j} \subset \im \Pi_{1}$ and $\Pi_{1} \neq 0$.

\begin{definition}\sl
A resonant state $u$ is an element of $\im \Pi_{1}$ which satisfies $( P - z ) u = 0$.
\end{definition}

In particular, resonant states are in $\CD_{\rm loc}$ but, in general, they are not in $\CH$. In the same spirit as in Theorem \ref{a1}, we obtain the following

\begin{theorem}\sl \label{a17}
Let $[ E_{0} , E_{1} ] \subset ] 0 , + \infty [$. There exists $a_{0} > R_{0}$ such that, for all $a_{0} < a < b$, $M > 0$ and $\varphi \in C^{\infty}_{0} ( \R^{n} )$, there exists $C > 0$ such that
\begin{equation} \label{a23}
\Vert \varphi u \Vert \leq C \sqrt{\frac{h}{\vert \im z \vert}} e^{C \vert \im z \vert / h} \big\Vert \one_{\CC_{a,b}} u \big\Vert ,
\end{equation}
for any resonant state $u$ associated to a resonance $z \in [ E_{0} , E_{1} ] - i [ 0 , M h \vert \ln h \vert ]$ and $h$ small enough.
\end{theorem}

Thus, this theorem gives a lower bound of the resonant states on the ring $\CC_{a,b}$. In a certain sense, it can be seen as an effective unique continuation result for the resonant states. However, we not consider the behavior at infinity of the resonant states.

\begin{remark}\sl \label{a18}
$i)$ Note that, under some assumptions and for resonances satisfying $\vert \im z \vert \lesssim h$, Stefanov \cite{St02_01} and Michel and the first author \cite{BoMi04_01} have shown that 
\begin{equation*}
\big\Vert \one_{\CC_{a,b}} u \big\Vert \lesssim \sqrt{\frac{\vert \im z \vert}{h}} \big\Vert \one_{B (b)} u \big\Vert .
\end{equation*}
Thus, the estimate given in Theorem \ref{a17} is sharp in this case.

$ii)$ Note also that one can use the known results concerning the resonant states to refine Theorem \ref{a17}. For instance, it is known that the resonant states are outgoing. This means that they vanish microlocally in the incoming region
\begin{equation*}
\Gamma_{-} ( \re z ) = \big\{ (x, \xi ) \in q^{-1} ( \re z ) ; \ \exp ( t H_q )( x , \xi ) \to \infty \text{ as } t \to -\infty \big\} . %\nrightarrow
\end{equation*}
We refer to Michel and the first author \cite{BoMi04_01} for a precise result. Thus, it can be possible, under some assumptions, to replace $u$ by $\Psi u$ in the right hand side of \eqref{a23} where $\Psi$ is a pseudodifferential operator which microlocalizes near the complement of the incoming region.

$iii)$ For Schr\"{o}dinger operators $P = - h^{2} \Delta + V ( x )$ and for simple resonances, Theorem \ref{a17} can be deduced from Theorem \ref{a1}. Indeed, letting the spectral parameter go to $z$ in Theorem \ref{a1}, we get
\begin{equation*}
\big\Vert \varphi \Pi_{1} \varphi \big\Vert \leq C \frac{h}{\vert \im z \vert} e^{C \vert \im z \vert / h} \big\Vert \one_{\CC_{a,b}} \Pi_{1} \one_{\CC_{a,b}} \big\Vert .
\end{equation*}
Therefore \eqref{a23} follows since, for Schr\"{o}dinger operators, we can write $\Pi_{1} = c u \< \bar{u} , \cdot \>$ for some $c \in \C \setminus \{ 0 \}$.

$iv)$ Theorem \ref{a17} shows that the resonant states associated to resonances at distance $h$ from the real axis cannot be localized near the trapped set to first order. More precisely, let $u (h)$ be a family of resonant states, with $\Vert u (h) \Vert_{B ( b )} = 1$, whose corresponding resonances $z (h)$ verify $h / A \leq - \im z ( h ) \leq A h$. Then, every semiclassical measure $\mu$ associated to $u (h)$ has the property
\begin{equation} \label{a41}
\mu ( \CC_{a,b} \times \R^{n} ) > 0 .
\end{equation}
Note that, for differential operators (i.e. $P = Q$), one could obtain \eqref{a41} by using the propagation properties of the semiclassical measures associated to the resonant states (see e.g. Theorem 4 of Nonnenmacher and Zworski \cite{NoZw09_01}).
\end{remark}

\begin{example}\rm
The estimates given in Theorem \ref{a17} and Remark \ref{a18} $i)$ are already known in the well in an island situation. In dimension $n = 1$ and at the bottom of the well, Helffer and Sj\"{o}strand \cite[Proposition 11.1]{HeSj86_01} (see also Harrell and Simon \cite{HaSi80_01}) have proved that the imaginary part of the first resonance satisfies
\begin{equation*}
\im z = - ( \alpha + o ( 1 ) ) h^{1 / 2} e^{- 2 S_{0} / h},
\end{equation*}
where $S_{0} > 0$ is the Agmon distance between the well and the sea and $\alpha \neq 0$ is explicit. On the other hand, the resonant state $u$ (normalized on $B ( b )$) verifies
\begin{equation*}
\big\Vert \one_{B ( b )} u \big\Vert = 1 \qquad \text{ and } \qquad \big\Vert \one_{\CC_{a,b}} u \big\Vert = ( \beta + o (1) ) h^{- 1 / 4} e^{- S_{0} / h} ,
\end{equation*}
with $\beta \neq 0$. This is in agreement with Theorem \ref{a17} and Remark \ref{a18} $i)$.

Note that the well in an island situation in the multidimensional case has been treated in \cite[Theorem 10.12]{HeSj86_01}. We also refer to Fujii\'e, Lahmar-Benbernou and Martinez \cite{FuLaMa11_01} for potentials which are only $C^{\infty}$ in a compact set. In all these works, the authors prove precise asymptotics of the resonant states and they obtain the imaginary part of the resonances by a formula similar to \eqref{a19} which is used in the proof of Theorem \ref{a17}.
\end{example}

\begin{example}\rm
The resonant states have also been computed for barrier-top resonances. In \cite[Theorem 4.1]{BoFuRaZe09_01}, Fujii\'e, Ramond, Zerzeri and the first author have proved that, for simple resonances with $\vert \im z \vert \lesssim h$, the resonant states $u$ are classical Lagrangian distributions whose Lagrangian manifold $\Lambda_{+}$ is the stable outgoing Lagrangian manifold at the critical point. Moreover, the principal symbol of $u$ does not vanish almost everywhere on $\Lambda_{+}$.

In particular, since the spatial projection of $\Lambda_{+}$ is the whole space $\R^{n}$, we get
\begin{equation*}
\Vert \varphi u \Vert \lesssim \big\Vert \one_{\CC_{a,b}} u \big\Vert \lesssim \Vert \varphi u \Vert ,
\end{equation*}
for all $0 \neq \varphi \in C^{\infty}_{0} ( \R^{n} )$. On the other hand, in this context, the imaginary part of a resonance satisfies $\im z = - \lambda h + o ( h )$ where $\lambda \neq 0$ is given by the eigenvalues of the Hessian of the potential at its maximum. This is in agreement with Theorem \ref{a17} and Remark \ref{a18} $i)$.
\end{example}

By our arguments, we can also study the generalized resonant states.

\begin{definition}\sl
A generalized resonant state $u$ is an element of $\im \Pi_{1}$. The order of $u$ is the smallest integer $J \geq 1$ such that $( P - z )^{J} u = 0$.
\end{definition}

Note that, using the notations of \eqref{a51}, the order of a generalized resonant state is bounded by $N$ because $( P - z ) \Pi_{N} = 0$ and $( P - z ) \Pi_{j} = \Pi_{j + 1}$ for $1 \leq j \leq N - 1$. As a consequence of Theorem \ref{a17}, we have the following result on the generalized resonant states of bounded order.

\begin{proposition}\sl \label{a27}
Let $[ E_{0} , E_{1} ] \subset ] 0 , + \infty [$. There exists $a_{0} > R_{0}$ such that, for all $a_{0} < a < b$, $M > 0$, $J \in \N \setminus \{ 0 \}$ and $\varphi \in C^{\infty}_{0} ( \R^{n} )$, there exists $C > 0$ such that
\begin{equation*}
\Vert \varphi u \Vert \leq C \sqrt{\frac{h}{\vert \im z \vert}} e^{C \vert \im z \vert / h} \sum_{j = 0}^{J - 1} \frac{1}{\vert \im z \vert^{j}} \big\Vert \one_{\CC_{a,b}} ( P - z )^{j} u \big\Vert ,
\end{equation*}
for any generalized resonant state $u$ of order less than $J$ associated to a resonance $z \in [ E_{0} , E_{1} ] - i [ 0 , M h \vert \ln h \vert ]$ and $h$ small enough.
\end{proposition}

The rest of this paper is organized as follows. In Section \ref{s2}, we prove Theorem \ref{a13} by constructing an auxiliary non-trapping operator which coincides with $P$ at infinity. Section \ref{s3} is devoted to the proof of Theorem \ref{a1}. The main idea is to exploit the formula
\begin{equation*}
\< \chi ( P - z ) u , \chi u \> - \< \chi u , \chi ( P - z ) u \> = \big\< [ \chi^{2} , P ] u , u \big\> - 2 i \im z \Vert \chi u \Vert^{2} ,
\end{equation*}
which is generally used to compute imaginary parts of resonances (see e.g. Helffer and Sj\"{o}strand \cite[Page 155]{HeSj86_01}). Proposition \ref{a10} is proved in Section \ref{s4} by building a well-chosen quasimode. The estimates concerning the resonant states are obtained in Section \ref{s5} using ideas similar to those of Section \ref{s3}. In Section \ref{s6}, we apply our results to the case of obstacle scattering and we make the link with the work of Stoyanov and the second author \cite{PeSt09_01}. Finally, we give some basic properties of the generalized resonant states in Appendix \ref{s7}.

\section{Proof of Theorem \ref{a13}} \label{s2}

First, we construct a non-trapping operator by planing $Q$ in a large compact set. This idea has been used by Robert and Tamura \cite{RoTa87_01} (see also Bruneau and the second author \cite{BrPe00_01} for trapping situations) to estimate the weighted resolvent on the real axis in non-trapping situations. Secondly, we recall the standard estimate of the cut-off resolvent associated to this new auxiliary operator. Let $\tau , \nu \in C^{\infty} ( \R^{n} ; [ 0 , 1 ] )$ be such that
\begin{equation*}
\one_{B ( 1/2 )} \prec \tau \prec \one_{B ( 1 )} ,
\end{equation*}
and $\tau^{2} + \nu^{2} = 1$ on $\R^{n}$. For $a > 0$, we define
\begin{equation*}
R_{a} = \nu \Big( \frac{x}{a} \Big) Q \nu \Big( \frac{x}{a} \Big) - \tau \Big( \frac{x}{a} \Big) h^{2} \Delta \tau \Big( \frac{x}{a} \Big) ,
\end{equation*}
a differential operator of order $2$ whose semiclassical principal symbol is
\begin{equation*}
r_{a} (x , \xi ) = q ( x , \xi ) \nu^{2} \Big( \frac{x}{a} \Big) + \xi^{2} \tau^{2} \Big( \frac{x}{a} \Big) .
\end{equation*}
In particular, $\xi^{2} / C - C \leq r_{a} \leq C \xi^{2} + C$ uniformly for $a > 0$. Moreover, using the assumption \eqref{d1}, a direct computation yields
\begin{align}
\{ r_{a} , x \cdot \xi \} &= \{ \xi^{2} , x \cdot \xi \} + \Big\{ ( q - \xi^{2} ) \nu^{2} \Big( \frac{x}{a} \Big) , x \cdot \xi \Big\}   \nonumber  \\
&= 2 \xi^{2} + o_{a \to + \infty} ( \< \xi \>^{2} ) = 2 r_{a} + o_{a \to + \infty} ( \< \xi \>^{2} ) \geq E_{0} / 2 > 0 , \label{a21}
\end{align}
for $r_{a} ( x , \xi ) \in [ E_{0} / 2 , 2 E_{1} ]$ and $a > a_{0}$ with $a_{0} > R_{0}$ sufficiently large. This implies that, for $a > a_{0}$, the symbol $r_{a}(x, \xi)$ is non-trapping on $r_a^{-1}(E)$ for all energies $E$ lying in the interval $[ E_{0} / 2 , 2 E_{1} ]$. Then, we can apply a result of Nakamura, Stefanov and Zworski \cite{NaStZw03_01} (see also Martinez \cite{Ma02_01}) which yields the following resolvent estimate.

\begin{lemma}\sl \label{a9}
For all $j \in \N$, $s \in \R$, $M > 0$ and $\varphi \in C^{\infty}_{0} ( \R^{n} )$, there exists $C > 0$ such that
\begin{equation*}
\big\Vert \varphi ( R_{a} - z )^{- j} \varphi \big\Vert_{H^{s}_{h} \to H^{s + 2}_{h}} \leq C \frac{e^{C \vert \im z \vert / h}}{h^{j}} ,
\end{equation*}
for $z \in [ E_{0} , E_{1} ] - i [ 0 , M h \vert \ln h \vert ]$ and $h$ small enough. Here,
\begin{equation*}
H^{s}_{h} ( \R^{n} ) = \big\{ u \in \CS^{\prime} ( \R^{n} ) ; \ \< h D_x \>^{s} u \in L^2 ( \R^{n} ) \big\} ,
\end{equation*}
is the semiclassical Sobolev space equipped with the norm $\Vert u \Vert_{H^{s}_{h}} = \Vert \< h D_x \>^{s} u \Vert_{L^{2}}$.
\end{lemma}

\begin{proof}
Since the operator $R_{a}$ is non-trapping on the energies in $[ E_{0} / 2 , 2 E_{1} ]$, we have
\begin{equation*}
\big\Vert \varphi ( R_{a} - z )^{-1} \varphi \big\Vert_{L^{2} \to L^{2}} \leq C \frac{e^{C \vert \im z \vert / h}}{h} ,
\end{equation*}
for $z \in [ E_{0} , E_{1} ] - i [ 0 , M h \vert \ln h \vert ] + B (h)$. This estimate follows from Proposition 3.1 of Nakamura, Stefanov and Zworski \cite{NaStZw03_01} and \eqref{a34} for $\im z \leq 0$ and from the usual Mourre theory (see e.g. Vasy and Zworski \cite{VaZw00_01}) for $\im z > 0$. In particular, for $z \in [ E_{0} , E_{1} ] - i [ 0 , M h \vert \ln h \vert ]$, it yields
\begin{equation*}
\big\Vert \varphi ( R_{a} - \lambda )^{-1} \varphi \big\Vert_{L^{2} \to L^{2}} \leq C \frac{e^{C \vert \im z \vert / h}}{h} ,
\end{equation*}
for all $\lambda \in z + B ( h )$. Then, the Cauchy formula implies
\begin{equation*}
\varphi ( R_{a} - z )^{- j} \varphi = \frac{1}{(j - 1 ) !} \partial_{z}^{j - 1} \varphi ( R_{a} - z )^{-1} \varphi = \frac{1}{2 i \pi} \oint_{z + \partial B (h)} \varphi ( R_{a} - \lambda )^{-1} \varphi \frac{d \lambda}{( \lambda - z )^{j}} ,
\end{equation*}
and then
\begin{equation} \label{a42}
\big\Vert \varphi ( R_{a} - z )^{-j} \varphi \big\Vert_{L^{2} \to L^{2}} \leq C \frac{e^{C \vert \im z \vert / h}}{h^{j}} .
\end{equation}
It remains to bound this operator from $H^{s}_{h}$ to $H^{s + 2}_{h}$. Since $R_{a}$ is an elliptic differential operator of order $2$, we have
\begin{equation*}
\Vert u \Vert_{H^{2 k}_{h}} \simeq \big\Vert ( R_{a} + i )^{k} u \big\Vert_{L^{2}} ,
\end{equation*}
for all $k \in \Z$. Thus, performing multiple commutations between $R_{a} + i$ and $\varphi ( R_{a} - \lambda )^{-j} \varphi$ and using \eqref{a42}, a standard argument gives
\begin{equation*}
\big\Vert \varphi ( R_{a} - z )^{-j} \varphi \big\Vert_{H^{2 k}_{h} \to H^{2 k + 2}_{h}} \leq C_{k} \frac{e^{C \vert \im z \vert / h}}{h^{j}} ,
\end{equation*}
for all $k \in \Z$. And the lemma follows from an interpolation argument.
\end{proof}

We now prove Theorem \ref{a13}. Assume that $\one_{B ( a )} \prec \chi$ with $a > a_{0}$ where $a_{0} > R_{0}$ is given by Proposition \ref{a10} and Lemma \ref{a9}. Let $\chi_{1} , \chi_{2} \in C^{\infty}_{0} ( \R^{n} )$ be such that
\begin{equation} \label{a14}
\one_{B ( a )} \prec \chi_{1} \prec \chi_{2} \prec \chi .
\end{equation}
In particular, $P ( 1 - \chi_{\bullet} ) = R_{a} ( 1 - \chi_{\bullet} )$. For $\im z > 0$ and then for $z \in \Lambda_{\theta_{0}}$ by meromorphic extension, we can write
\begin{align}
\varphi ( P - z )^{-1} \varphi ={}& \varphi \one_{\R^{n} \setminus B ( R_{0} )} ( R_{a} - z )^{-1} ( 1 - \chi_{1} ) \varphi + \varphi \chi_{1} ( P - z )^{-1} \chi_{2} \varphi    \nonumber  \\
&+ \varphi \chi_{1} ( P - z )^{-1} [ P , \chi_{2} ] ( R_{a} - z )^{-1} \one_{\R^{n} \setminus B ( R_{0} )} \varphi   \nonumber  \\
&- \varphi \one_{\R^{n} \setminus B ( R_{0} )} ( R_{a} - z )^{-1} [ P , \chi_{1} ] ( P - z )^{-1} \chi_{2} \varphi  \nonumber  \\
&- \varphi \one_{\R^{n} \setminus B ( R_{0} )} ( R_{a} - z )^{-1} [ P , \chi_{1} ] ( P - z )^{-1} [ P , \chi_{2} ] ( R_{a} - z )^{-1} \one_{\R^{n} \setminus B ( R_{0} )} \varphi .  \label{a16}
\end{align}
To prove this identity for $\im z > 0$, the cut-off function $\varphi$ can be omitted and it is enough to expand the commutator $[ P , \chi_{2} ]$ and then the commutator $[ P , \chi_{1} ]$, and to use the formula $[ P , \chi_{\bullet} ] = ( P - z ) ( \chi_{\bullet} - 1 ) - ( \chi_{\bullet} - 1 ) ( P - z )$. The properties of the $\chi_{\bullet}$'s given in \eqref{a14} imply that
\begin{equation} \label{a15}
[ P , \chi_{\bullet} ] = \chi (x) \< h \nabla \> h \CO ( 1 ) \chi (x) ,
\end{equation}
where the $\CO ( 1 )$ denotes an operator bounded uniformly in $h$ on $L^{2} ( \R^{n} )$. Combining Lemma \ref{a9}, \eqref{a16} and \eqref{a15} (with its adjoint), we finally obtain
\begin{align}
\big\Vert \varphi ( P - z )^{-1} \varphi \big\Vert \lesssim{}& \big\Vert \varphi ( R_{a} - z )^{-1} \varphi \big\Vert + \big\Vert \chi ( P - z )^{-1} \chi \big\Vert   \nonumber  \\
&+ h \big\Vert \chi ( P - z )^{-1} \chi \big\Vert \big\Vert \chi ( R_{a} - z )^{-1} \varphi \big\Vert_{L^{2} \to H^{1}_{h}} \nonumber  \\
&+ h \big\Vert \chi ( P - z )^{-1} \chi \big\Vert \big\Vert \varphi ( R_{a} - z )^{-1} \chi \big\Vert_{H^{-1}_{h} \to L^{2}} \nonumber  \\
&+ h^{2} \big\Vert \chi ( P - z )^{-1} \chi \big\Vert \big\Vert \varphi ( R_{a} - z )^{-1} \chi \big\Vert_{H^{-1}_{h} \to L^{2}} \big\Vert \chi ( R_{a} - z )^{-1} \varphi \big\Vert_{L^{2} \to H^{1}_{h}}  \nonumber  \\
\lesssim{}& \frac{e^{C \vert \im z \vert / h}}{h} + \big\Vert \chi ( P - z )^{-1} \chi \big\Vert \Big( 1 + e^{2 C \vert \im z \vert / h} \Big) .
\end{align}
To complete the proof of Theorem \ref{a13}, it is enough to use Proposition \ref{a10}.

\section{Proof of Theorem \ref{a1}} \label{s3}

We will first estimate $\chi_{1} ( P - z )^{- 1} \chi_{1}$ for a particular cut-off function $\chi_{1}$ adapted to the ring $\CC_{a,b}$ and then apply Theorem \ref{a13} to estimate $\varphi ( P - z )^{- 1} \varphi$ for all $\varphi \in C^{\infty}_{0} ( \R^{n} )$. Let $\chi_{1} , \chi_{2} , \chi_{3} , \chi_{4} \in C^{\infty}_{0} ( \R^{n} )$ be such that $\one_{B ( a )} \prec \chi_{1} \prec \chi_{2} \prec \chi_{3} \prec \chi_{4} \prec \one_{B ( b )}$. We also consider $\psi_{1} , \psi_{2} , \psi_{3} , \psi_{4} \in C^{\infty}_{0} ( \R^{n} )$ such that $\nabla \chi_{1} \prec \psi_{1} \prec \psi_{2} \prec \chi_{2} \one_{\CC_{a,b}}$, $\nabla \chi_{3} \prec \psi_{3} \prec \psi_{4} \prec \chi_{4} \one_{\CC_{a,b}}$ and $\chi_{2} \psi_{4} = 0$. We begin with the following estimates.

\begin{lemma}\sl
For $f \in \CH_{\rm comp}$ and $z \in \Lambda_{\theta_{0}}$ with $\vert \im z \vert \leq 1$, we have
\begin{gather}
\big\Vert \chi_{3} ( P - z )^{- 1} f \big\Vert^{2} \lesssim \frac{1}{\vert \im z \vert^{2}} \Vert \chi_{4} f \Vert^{2} + \frac{h}{\vert \im z \vert} \big\Vert \psi_{4} ( P - z )^{- 1} f \big\Vert^2 ,    \label{a6}  \\
\big\Vert \chi_{1} {( P - z )^{- 1}}^{*} f \big\Vert^{2} \lesssim \frac{1}{\vert \im z \vert^{2}} \Vert \chi_{2} f \Vert^{2} + \frac{h}{\vert \im z \vert} \big\Vert \psi_{2} {( P - z )^{- 1}}^{*} f \big\Vert^2 .  \label{a7}
\end{gather}
\end{lemma}

\begin{proof}
For $u \in \CD_{\rm loc}$, we have
\begin{equation*}
\< \chi_{3} (P - z) u, \chi_{3} u \> - \< \chi_{3} u , \chi_{3} ( P - z ) u \> = \< [ \chi_{3}^{2} , P ] u, u \> - 2 i \im z \Vert \chi_{3} u \Vert^{2} .
\end{equation*}
Taking $u = ( P - z )^{-1} f$ yields
\begin{align}
\vert \im z \vert \big\Vert \chi_{3} ( P - z )^{-1} f \big\Vert^{2} \lesssim{}& \big\Vert \chi_{3} ( P - z )^{-1} f \big\Vert \Vert \chi_{3} f \Vert   \nonumber \\
&+ \big\Vert [ \chi_{3}^{2} , P ] ( P - z )^{-1} f \big\Vert \big\Vert \psi_{4} ( P - z )^{-1} f \big\Vert .   \label{a4}
\end{align}
Moreover, combining \eqref{a15}, the ellipticity of $P$ and the properties of the support of the cut-off functions, we obtain
\begin{align}
[ \chi_{3}^{2} , P ] ( P - z )^{-1} f ={}& [ \chi_{3}^{2} , P ] ( P + i )^{-1} ( P + i ) \psi_{3} ( P - z )^{-1} f  \nonumber \\
={}& [ \chi_{3}^{2} , P ] ( P + i )^{-1} \big( \psi_{3} ( P + i ) + [ P , \psi_{3} ] \big) ( P - z )^{-1} f  \nonumber \\
={}& [ \chi_{3}^{2} , P ] ( P + i )^{-1} \psi_{3} f  \nonumber \\
&+ [ \chi_{3}^{2} , P ] ( P + i )^{-1} \big( ( i + z ) \psi_{3} + [ P , \psi_{3} ] \big) ( P - z )^{-1} f  \nonumber \\
={}& \CO ( h ) \Vert \chi_{4} f \Vert + \CO ( h ) \big\Vert \psi_{4} ( P - z )^{-1} f \big\Vert .   \label{a5}
\end{align}
Combining \eqref{a4} and \eqref{a5}, we obtain
\begin{align*}
\vert \im z \vert \big\Vert \chi_{3} ( P - z )^{-1} f \big\Vert^{2} \leq{}& \frac{\vert \im z \vert}{2} \big\Vert \chi_{3} ( P - z )^{-1} f \big\Vert^{2} + \frac{C}{\vert \im z \vert} \Vert \chi_{3} f \Vert^{2}    \\
&+ C h\Vert \chi_{4} f \Vert^{2} + C h \big\Vert \psi_{4} ( P - z )^{-1} f \big\Vert^{2}  \\
\lesssim{}& \frac{1}{\vert \im z \vert} \Vert \chi_{4} f \Vert^{2} + h \big\Vert \psi_{4} ( P - z )^{-1} f \big\Vert^{2} .
\end{align*}
This implies \eqref{a6}. The estimate for the adjoint operator \eqref{a7} can be proved by the same argument using ${( P - z )^{- 1}}^{*} = ( P - \bar{z} )^{-1}$.
\end{proof}

We can now prove Theorem \ref{a1}. Recall that, for simplicity, we use the notation $\Vert \cdot \Vert$ for the norm of the space $\CH$ and the operator norm on $\CH$. To be more precise, in the rest of this section, $\Vert \cdot \Vert$ denotes the norm of $\CH$ only when $f$ or $u$ appears in the expression. From \eqref{a6}, we can write
\begin{align*}
\big\Vert \chi_{1} ( P - z )^{- 1} \chi_{1} f \big\Vert^{2} &\leq \big\Vert \chi_{3} ( P - z )^{- 1} \chi_{1} f \big\Vert^{2}     \\
&\lesssim \frac{1}{\vert \im z \vert^{2}} \Vert \chi_{4} \chi_{1} f \Vert^{2} + \frac{h}{\vert \im z \vert} \big\Vert \psi_{4} ( P - z )^{- 1} \chi_{1} f \big\Vert^{2}    \\
&\leq \frac{1}{\vert \im z \vert^{2}} \Vert f \Vert^{2} + \frac{h}{\vert \im z \vert} \big\Vert \chi_{1} {( P - z )^{- 1}}^{*} \psi_{4} \big\Vert^{2} \Vert f \Vert^{2} .
\end{align*}
Using now \eqref{a7} and $\chi_{2} \psi_{4} = 0$, we get
\begin{align*}
\big\Vert \chi_{1} ( P - z )^{- 1} \chi_{1} f \big\Vert^{2} &\lesssim \frac{1}{\vert \im z \vert^{2}} \Vert f \Vert^{2} + \frac{h}{\vert \im z \vert} \Vert f \Vert^{2} \sup_{\Vert u \Vert = 1} \big\Vert \chi_{1} {( P - z )^{- 1}}^{*} \psi_{4} u \big\Vert^{2}    \\
&\lesssim \frac{1}{\vert \im z \vert^{2}} \Vert f \Vert^{2} + \frac{h}{\vert \im z \vert} \Vert f \Vert^{2} \sup_{\Vert u \Vert = 1} \Big( \frac{1}{\vert \im z \vert^{2}} \Vert \chi_{2} \psi_{4} u \Vert^{2}   \\
&\qquad \qquad \qquad \qquad \qquad \qquad + \frac{h}{\vert \im z \vert} \big\Vert \psi_{2} {( P - z )^{- 1}}^{*} \psi_{4} u \big\Vert^{2} \Big)    \\
&=\frac{1}{\vert \im z \vert^{2}} \Vert f \Vert^{2} + \frac{h^{2}}{\vert \im z \vert^{2}} \big\Vert \psi_{4} ( P - z )^{- 1} \psi_{2} \big\Vert^{2} \Vert f \Vert^{2} .   \\
\end{align*}
Combining with $\psi_{\bullet} \prec \one_{\CC_{a,b}}$ yields
\begin{equation*}
\big\Vert \chi_{1} ( P - z )^{- 1} \chi_{1} \big\Vert \lesssim \frac{1}{\vert \im z \vert} + \frac{h}{\vert \im z \vert} \big\Vert \one_{\CC_{a,b}} ( P - z )^{- 1} \one_{\CC_{a,b}} \big\Vert .
\end{equation*}
We now apply Theorem \ref{a13} and assume that $a \geq a_{0}$. Since $\one_{B ( a_{0} )} \prec \chi_{1}$, Theorem \ref{a13} together with the previous estimate gives
\begin{equation} \label{a8}
\big\Vert \varphi ( P - z )^{- 1} \varphi \big\Vert \lesssim \frac{h}{\vert \im z \vert} e^{C \vert \im z \vert/ h} \big( \big\Vert \one_{\CC_{a,b}} ( P - z )^{- 1} \one_{\CC_{a,b}} \big\Vert + h^{-1} \big) .
\end{equation}
To conclude the proof of Theorem \ref{a1}, it is enough to apply Proposition \ref{a10}.

\section{Proof of Proposition \ref{a10}} \label{s4}

To prove this result, we construct a quasimode of order $h$. Since the semiclassical principal symbol $q (x , \xi )$ of $Q$ converges to $\xi^{2}$ at infinity, there exists $a_{0} > R_{0}$ such that, for all $\vert x \vert \geq a_{0}$, we have $q ( x , 0 ) \leq E_{0} / 2$. Let now $\varphi \in C^{\infty}_{0} ( \R^{n} )$ and $\vert x_{0} \vert \geq a_{0}$ be such that $\varphi ( x_{0} ) \neq 0$.

Using $q ( x_{0} , 0 ) < E_{0}$ and the form of $q ( x_{0} , \cdot )$ given in \eqref{d2}, one can construct $\xi_{0} ( \lambda ) \in C^{\infty}$ such that $q ( x_{0} , \xi_{0} ( \lambda ) ) = \lambda$ and $\partial_{\xi_{1}} q ( x_{0} , \xi_{0} ( \lambda ) ) \neq 0$ for all $\lambda \in [ E_{0} , E_{1} ]$. Solving the Hamilton--Jacobi equation by the usual method (see e.g. Dimassi and Sj\"{o}strand \cite[Theorem 1.5]{DiSj99_01}), there exists a phase function $\psi ( x , \lambda ) \in C^{\infty}$ defined for $x$ in a neighborhood of $x_{0}$ and for $\lambda \in [ E_{0} , E_{1} ]$, and such that
\begin{equation*}
q ( x , \nabla_{x} \psi ( x , \lambda ) ) = \lambda ,
\end{equation*}
for all $\lambda \in [ E_{0} , E_{1} ]$. Let now
\begin{equation*}
u ( x , z ) = \chi ( x ) e^{i \psi (x , \re z )/ h} ,
\end{equation*}
where $0 \neq \chi \in C^{\infty}_{0} ( \R^{n} )$ is supported in the intersection of $W = \{ x ; \ \vert \varphi ( x ) \vert \geq \vert \varphi ( x_{0} ) \vert / 2 \}$ and the set where $\psi$ is defined.

Let $P_{\theta}$ be the operator $P$ distorted outside the support of $\varphi$ by a fixed angle $0 < \theta \leq \theta_{0}$ large enough. A standard computation by the method of stationary phase gives
\begin{align}
( P_{\theta} - z ) u &= ( P - z ) u = ( Q - z ) u   \nonumber  \\
&= ( \op ( q ) - \re z ) u + h Q_{1} u - i \im z \, u    \nonumber  \\
&= \big( q ( x , \nabla_{x} \psi ( x , \re z ) ) - \re z \big) u + \CO ( h + \vert \im z \vert ) \nonumber  \\
&= \CO ( h + \vert \im z \vert ) .    \label{a20}
\end{align}
where $\op ( q )$ is any semiclassical quantization of $q$ and $Q_{1}$ is a $h$-differential operator of order two with coefficients uniformly bounded with respect to $h$. Using that $( P_{\theta} - z ) u = ( Q - z ) u$ is supported in $W$, we so can write
\begin{equation*}
( P_{\theta} - z ) u = \varphi v ,
\end{equation*}
where
\begin{equation} \label{d3}
\Vert v ( x , z ) \Vert \lesssim h + \vert \im z \vert .
\end{equation}
Then, using that $(P_{\theta} - z)^{-1}$ is invertible and the equality \eqref{a34}, we get
\begin{equation*}
\varphi u = \varphi ( P_{\theta} - z )^{-1} \varphi v = \varphi ( P - z )^{-1} \varphi v .
\end{equation*}
Finally, combining the previous equation with \eqref{d3} and $\Vert \varphi u \Vert \gtrsim 1$, we obtain
\begin{equation*}
\big\Vert \varphi ( P - z )^{-1} \varphi \big\Vert \gtrsim \frac{1}{h + \vert \im z \vert} \geq h^{- 1} e^{- \vert \im z \vert / h} ,
\end{equation*}
and the proposition follows.

\section{Estimates for the resonant states} \label{s5}

In this part, we prove the estimates for the (generalized) resonant states given is Section \ref{s1}.

\begin{proof}[Proof of Theorem \ref{a17}]
Choose cut-off functions $\chi , \widetilde{\chi} \in C^{\infty}_{0} ( \R^{n} )$ so that
\begin{equation} \label{a30}
\one_{B ( a )} \prec \chi \prec \one_{B ( b )} \qquad \text{and} \qquad \nabla \chi \prec \widetilde{\chi} \prec \one_{\CC_{a,b}}.
\end{equation}
Let $u$ be a resonant state associated to a resonance $z \in [ E_{0} , E_{1} ] - i [ 0 , M h \vert \ln h \vert ]$. We first estimate $\chi u$. Since $u \in \CD_{\rm loc}$ and $( P - z ) u = 0$, we have
\begin{align}
0 &= \< \chi ( P - z ) u , \chi u \> - \< \chi u , \chi ( P - z ) u \>   \nonumber \\
&= \big\< [ \chi^{2} , P ] u , u \big\> - 2 i \im z \Vert \chi u \Vert^{2} .    \label{a19}
\end{align}
Thus we obtain
\begin{equation} \label{a25}
\Vert \chi u \Vert^2 \leq \frac{1}{2 \vert \im z \vert} \big\vert \big\< [ \chi^{2} , P ] \widetilde{\chi} u , \widetilde{\chi} u \big\> \big\vert .
\end{equation}
To estimate the action of $[ \chi^{2} , P ]$ on $\widetilde{\chi} u$, we write
\begin{align}
[ \chi^{2} , P ] \widetilde{\chi} u &= [ \chi^{2} , P ] ( P + i )^{-1} ( P + i ) \widetilde{\chi} u  \nonumber \\
&= [ \chi^{2} , P ] ( P + i )^{-1} \big( \widetilde{\chi} ( P + i ) u + [ P , \widetilde{\chi} ] u \big)     \nonumber \\
&= [ \chi^{2} , P ] ( P + i )^{-1} \big( \widetilde{\chi} ( z + i ) \one_{\CC_{a,b}} u + [ P , \widetilde{\chi} ] \one_{\CC_{a,b}} u \big) .  \label{a26}
\end{align}
The operator $[ \chi^{2} , P ] ( P + i )^{-1} \widetilde{\chi} : L^2 \rightarrow L^2$ is bounded by $\CO (h)$, while the operator
\begin{equation*}
[ \chi^{2} , P ] ( P + i )^{-1} [ P , \widetilde{\chi} ] : L^2 \rightarrow L^2 ,
\end{equation*}
is bounded by $\CO ( h^{2} )$. Thus, combining \eqref{a25} and \eqref{a26}, we deduce
\begin{equation} \label{a24}
\Vert \chi u \Vert \leq C \sqrt{\frac{h}{\vert \im z \vert}} \big\Vert \one_{\CC_{a,b}} u \big\Vert .
\end{equation}

We now estimate $\varphi u$ for all $\varphi \in C^{\infty}_{0} ( \R^{n} )$. Let $P_{\theta}$ (resp. $R_{a , \theta}$) be a complex distortion of $P$ (resp. of $R_{a}$ which is defined in Section \ref{s2}) by a fixed angle $0 < \theta \leq \theta_{0}$. We also assume that the scaling occurs only outside of $\supp \varphi \cup B ( b )$. Then, from Lemma \ref{a39}, there exists $u_{\theta} \in \CD$ such that $( P_{\theta} - z ) u_{\theta} = 0$,
\begin{equation} \label{a32}
\one_{B ( b )} u_{\theta} = \one_{B ( b )} u \qquad \text{and} \qquad \varphi u_{\theta} = \varphi u .
\end{equation}
On the other hand, the definition of $R_{a}$ and $\one_{B ( a )} \prec \chi$ imply $R_{a , \theta} ( 1 - \chi ) = P_{\theta} ( 1 - \chi )$. Thus, we can write
\begin{equation*}
( R_{a , \theta} - z ) ( 1 - \chi ) u_{\theta} = ( P_{\theta} - z ) ( 1 - \chi ) u_{\theta} = - [ P , \chi ] u_{\theta} .
\end{equation*}
This yields
\begin{equation*}
( 1 - \chi ) u_{\theta} = - ( 1 - \widehat{\chi} ) ( R_{a , \theta} - z )^{-1} [ P, \chi ] u_{\theta} ,
\end{equation*}
where $\widehat{\chi} \in C^{\infty}_{0} ( \R^{n} )$, with $\one_{B ( R_{0} )} \prec \widehat{\chi} \prec \chi$, is an artificial cut-off function used to identify $( 1 - \widehat{\chi} ) \CH$ and $( 1 - \widehat{\chi} ) L^{2}$. Finally, we get
\begin{align}
\varphi u &= \varphi u_{\theta} = \varphi \chi u_{\theta} - \varphi ( 1 - \widehat{\chi} ) ( R_{a , \theta} - z )^{-1} [ P , \chi ] u_{\theta}  \nonumber \\
&= \varphi \chi u - ( 1 - \widehat{\chi} ) \varphi ( R_{a} - z )^{-1} \widetilde{\chi} [ P , \chi ] \one_{\CC_{a,b}} u .   \label{a31}
\end{align}
To complete the proof of Theorem \ref{a17}, it is enough to use \eqref{a24} and
\begin{align*}
\big\Vert ( 1 - \widehat{\chi} ) \varphi ( R_{a} - z )^{-1} \widetilde{\chi} [ P , \chi ] \big\Vert_{\CH \to \CH} &\lesssim \big\Vert \varphi ( R_{a} - z )^{-1} \widetilde{\chi} \big\Vert_{H_{h}^{-1} \to L^{2}} \big\Vert [ P , \chi ] \big\Vert_{L^{2} \to H_{h}^{-1}}   \\
&\lesssim \frac{e^{C \vert \im z \vert / h}}{h} \times h \leq C \sqrt{\frac{h}{\vert \im z \vert}} e^{C \vert \im z \vert / h} ,
\end{align*}
which follows from Lemma \ref{a9}.
\end{proof}

\begin{proof}[Proof of Proposition \ref{a27}]
We will prove this result by induction over the order $J$ of the generalized resonant state $u$. For $J = 1$, Proposition \ref{a27} is a direct consequence of Theorem \ref{a17}. Now assume that Proposition \ref{a27} holds true for generalized resonant states of order less than $J - 1$ for some $J \geq 2$. Let $u$ be a generalized resonant state of order $J$. Following the analysis of \eqref{a19}, we have
\begin{equation*}
\< \chi ( P - z ) u , \chi u \> - \< \chi u , \chi ( P - z ) u \> = \big\< [ \chi^{2} , P ] u , u \big\> - 2 i \im z \Vert \chi u \Vert^{2} ,
\end{equation*}
which implies
\begin{align}
\Vert \chi u \Vert^{2} &\leq \frac{1}{2 \vert \im z \vert} \big\vert \big\< [ \chi^{2} , P ] u , u \big\> \big\vert + \frac{1}{\vert \im z \vert} \Vert \chi u \Vert \big\Vert \chi ( P - z ) u \big\Vert    \nonumber \\
&\leq \frac{1}{2 \vert \im z \vert} \big\vert \big\< [ \chi^{2} , P ] u , u \big\> \big\vert + \frac{1}{2} \Vert \chi u \Vert^{2} + \frac{1}{2 \vert \im z \vert^{2}} \big\Vert \chi ( P - z ) u \big\Vert^{2}    \nonumber \\
&\leq \frac{1}{\vert \im z \vert} \big\vert \big\< [ \chi^{2} , P ] \widetilde{\chi} u , \widetilde{\chi} u \big\> \big\vert + \frac{1}{\vert \im z \vert^{2}} \big\Vert \chi ( P - z ) u \big\Vert^{2} .    \label{a28}
\end{align}
As in \eqref{a26}, we can write
\begin{align*}
[ \chi^{2} , P ] \widetilde{\chi} u ={}& [ \chi^{2} , P ] ( P + i )^{-1} ( P + i ) \widetilde{\chi} u   \\
={}& [ \chi^{2} , P ] ( P + i )^{-1} \widetilde{\chi} ( P - z ) u    \\
&+ [ \chi^{2} , P ] ( P + i )^{-1} \big( \widetilde{\chi} ( z + i ) \one_{\CC_{a,b}} u + [ P , \widetilde{\chi} ] \one_{\CC_{a,b}} u \big) ,
\end{align*}
which yields
\begin{equation*}
\big\Vert [ \chi^{2} , P ] \widetilde{\chi} u \big\Vert \lesssim h \big\Vert \one_{\CC_{a,b}} ( P - z ) u \big\Vert + h \big\Vert \one_{\CC_{a,b}} u \big\Vert .
\end{equation*}
Then, \eqref{a28} becomes
\begin{equation*}
\Vert \chi u \Vert \lesssim \sqrt{\frac{h}{\vert \im z \vert}} \big\Vert \one_{\CC_{a,b}} u \big\Vert + \sqrt{\frac{h}{\vert \im z \vert}} \big\Vert \one_{\CC_{a,b}} ( P - z ) u \big\Vert + \frac{1}{\vert \im z \vert} \big\Vert \chi ( P - z ) u \big\Vert .
\end{equation*}
Now we remark that $( P - z ) u \in \Pi_{1}$ is a generalized resonant state whose order is $J - 1$. Then, applying the recurrence assumption, the previous equation gives
\begin{equation} \label{a29}
\Vert \chi u \Vert \lesssim \sqrt{\frac{h}{\vert \im z \vert}} e^{C \vert \im z \vert / h} \sum_{j = 0}^{J - 1} \frac{1}{\vert \im z \vert^{j}} \big\Vert \one_{\CC_{a,b}} ( P - z )^{j} u \big\Vert .
\end{equation}

Next will now obtain a formula similar to \eqref{a31} to control $\varphi u$ for $\varphi \in C^{\infty}_{0} ( \R^{n} )$. As in \eqref{a32}, let $P_{\theta}$ (resp. $R_{a , \theta}$) be a complex distortion of $P$ (resp. $R_{a}$) by a fixed angle $0 < \theta \leq \theta_{0}$. Assume also that the scaling occurs only outside of $\supp \varphi \cup B ( b )$. Then, from Lemma \ref{a39}, there exists $u_{\theta} \in \CD^{J}$ such that $( P_{\theta} - z )^{J} u_{\theta} = 0$,
\begin{equation*}
\one_{B ( b )} u_{\theta} = \one_{B ( b )} u \qquad \text{and} \qquad \varphi u_{\theta} = \varphi u .
\end{equation*}
We also have $R_{a , \theta} ( 1 - \chi ) = P_{\theta} ( 1 - \chi )$. A direct computation gives
\begin{equation*}
( R_{a , \theta} - z )^{J} ( 1 - \chi ) u_{\theta} = ( P_{\theta} - z )^{J} ( 1 - \chi ) u_{\theta} = - \sum_{j = 0}^{J - 1} \big( \ad_{P}^{J - j} \chi \big) ( P - z )^{j} u ,
\end{equation*}
where $\ad_{P}^{0} \chi = \chi$ and $\ad_{P}^{j+1} \chi = [ P , \ad_{P}^{j} \chi ]$. Thus, mimicking the proof of \eqref{a31}, we get
\begin{equation*}
\varphi u = \varphi \chi u - ( 1 - \widehat{\chi} ) \varphi ( R_{a} - z )^{- J} \widetilde{\chi} \sum_{j = 0}^{J - 1} \big( \ad_{P}^{J - j} \chi \big) \one_{\CC_{a,b}} ( P - z )^{j} u .
\end{equation*}
Using \eqref{a29}, Lemma \ref{a9} and $\Vert \ad_{P}^{j} \chi \Vert_{H^{s}_{h} \to H^{s - j}_{h}} = \CO ( h^{j} )$, the previous equation gives
\begin{align*}
\Vert \varphi u \Vert &\lesssim \sqrt{\frac{h}{\vert \im z \vert}} e^{C \vert \im z \vert / h} \sum_{j = 0}^{J - 1} \frac{1}{\vert \im z \vert^{j}} \big\Vert \one_{\CC_{a,b}} ( P - z )^{j} u \big\Vert + e^{C \vert \im z \vert / h} \sum_{j = 0}^{J - 1} \frac{1}{h^{j}} \big\Vert \one_{\CC_{a,b}} ( P - z )^{j} u \big\Vert   \\
&\lesssim \sqrt{\frac{h}{\vert \im z \vert}} e^{C \vert \im z \vert / h} \sum_{j = 0}^{J - 1} \frac{1}{\vert \im z \vert^{j}} \big\Vert \one_{\CC_{a,b}} ( P - z )^{j} u \big\Vert ,
\end{align*}
since $h^{-1} \lesssim \vert \im z \vert^{- 1} e^{\vert \im z \vert / h}$. Thus Proposition \ref{a27} holds for generalized resonant states of order $J$ and the proof is complete.
\end{proof}

\section{Scattering by obstacles} \label{s6}

Let $K \subset \{ x \in \R^n ; \ \vert x \vert \leq R_{0} \}$, $n \geq 2$, be a bounded domain with smooth boundary such that $\Omega = \R^n \setminus \overline{K}$ is connected. Let $- \Delta_D$ be the Dirichlet Laplacian in $\Omega$ which is a self-adjoint operator on $\CH = L^{2}( \Omega )$ with domain $\CD = H_{0}^{1} ( \Omega ) \cap H^{2} ( \Omega )$. For $\im \lambda > 0$ the resolvent $( - \Delta_D - \lambda^2)^{-1}$ is a bounded operator from $\CH$ to $\CD$ and, for all $\varphi \in C_{0}^{\infty} ( \Omega )$, the cut-off resolvent $\varphi ( - \Delta_D - \lambda^2)^{-1} \varphi$ admits a meromorphic continuation in $\C$ for $n$ odd and in $\C \setminus i \R^{-}$ for $n$ even. For non-trapping perturbations, we have an estimate
\begin{equation*}
\big\Vert \varphi ( - \Delta_{D} - \lambda^{2} )^{- 1} \varphi \big\Vert \lesssim \< \lambda \>^{- 1} ,
\end{equation*}
for $\lambda \in \R$, $\vert \lambda \vert \geq 1$, while for trapping perturbations and $\lambda \in \R$, $\vert \lambda \vert \geq 1$ this cut-off resolvent is bounded by $e^{C \vert \lambda \vert}$ (see Burq \cite{Bu98_01}).

Since we will use the Lax--Phillips theory \cite{LaPh89_01}, we consider in $\Omega$ the wave equation
\begin{equation} \label{a43}
\partial_{t}^{2} u - \Delta_{D} u = 0 ,
\end{equation}
with Dirichlet boundary condition on $\partial \Omega$. Let $H_D ( \Omega )$ be the closure of $C^{\infty}_{0} ( \Omega )$ for the norm $\Vert \nabla \cdot \Vert_{L^{2} ( \Omega )}$. We introduce the energy space $H = H_D ( \Omega ) \oplus L^2 ( \Omega )$ and the unitary group $e^{-i t G} : H \longrightarrow H$ with generator $- i G$, where
\begin{equation*}
G = i \left( \begin{matrix} 0 & I \\ \Delta_{D} & 0 \end{matrix} \right) ,
\end{equation*}
is a self-adjoint operator on $H$ (see Lax and Phillips \cite{LaPh89_01}). As usual, the solutions of \eqref{a43} are given by
\begin{equation}
\left( \begin{matrix} u (t) \\ \partial_{t} u (t) \end{matrix} \right) = e^{-i t G} \left( \begin{matrix} u (0) \\ \partial_{t} u (0) \end{matrix} \right) .
\end{equation}

To apply the results for semiclassical operators established in Section \ref{s1}, we consider the scaling $\lambda = \frac{\sqrt{z}}{h}$ and write
\begin{equation} \label{a45}
( - \Delta_{D} - \lambda^{2} )^{- 1} = h^{2} ( P - z )^{- 1} ,
\end{equation}
where $P = - h^{2} \Delta_{D}$ satisfies the general assumptions of Section \ref{s1}. We want to estimate the cut-off resolvent of $- \Delta_{D}$ in the region
\begin{equation*}
\CS = \big\{ \lambda \in \C ; \ \re \lambda \geq 1 \text{ and } 0 \geq \im \lambda \geq - M \ln ( \re \lambda ) \big\} .
\end{equation*}
It is then enough to consider the situation
\begin{equation*}
\lambda \in \CS_{h} = \big[ h^{-1} , 2 h^{-1} \big] - i \big[ 0 , M ( \vert \ln h \vert + \ln 2 ) \big] ,
\end{equation*}
since the union of $\CS_{h}$ over $0 < h \leq 1$ covers $\CS$. For $\lambda \in \CS_{h}$, we have
\begin{equation*}
\sqrt{z} \in [ 1 , 2 ] - i \big[ 0 , h M ( \vert \ln h \vert + \ln 2 ) \big] ,
\end{equation*}
and finally
\begin{equation*}
z \in [1/2 , 4 ] - i [ 0 , 5 M h \vert \ln h \vert ] ,
\end{equation*}
for $h$ small enough. Applying Theorem \ref{a1} in this region to the operator $P$ and using the relation \eqref{a45}, we obtain, for $\lambda \in \CS_{h}$ with $h$ small enough,
\begin{align*}
\big\Vert \varphi ( - \Delta_{D} - \lambda^{2} )^{- 1} \varphi \big\Vert &= \big\Vert h^{2} \varphi ( P - z )^{- 1} \varphi \big\Vert  \\
&\leq C \frac{h}{\vert \im z \vert} e^{C \vert \im z \vert / h} \big\Vert h^{2} \one_{\CC_{a,b}} ( P - z )^{-1} \one_{\CC_{a,b}} \big\Vert   \\
&\leq C \frac{e^{C \vert \im \lambda \vert}}{\vert \im \lambda \vert} \big\Vert \one_{\CC_{a,b}} ( - \Delta_{D} - \lambda^{2} )^{- 1} \one_{\CC_{a,b}} \big\Vert ,
\end{align*}
since $\vert \im z \vert /h$ behaves like $\vert \im \lambda \vert$ in $\CS_{h}$. Note also that such relation holds true in any compact set (with a constant $C$ depending on the compact set). This follows from Corollary \ref{a37} near the resonances and from the fact that $\one_{\CC_{a,b}} ( - \Delta_{D} - \lambda^{2} )^{- 1} \one_{\CC_{a,b}}$ does not vanish (because $\chi = ( - \Delta_{D} - \lambda^{2} )^{- 1} ( - \Delta_{D} - \lambda^{2} ) \chi$) away from the resonances. Summing up, we have proved the following

\begin{theorem}\sl
There exists $a_{0} > R_{0}$ such that, for all $a_{0} < a < b$, $M > 0$ and $\varphi \in C^{\infty}_{0} ( \R^{n} )$, there exists $C > 0$ such that
\begin{equation} \label{a44}
\big\Vert \varphi ( - \Delta_{D} - \lambda^{2} )^{- 1} \varphi \big\Vert \leq C \frac{e^{C \vert \im \lambda \vert}}{\vert \im \lambda \vert} \big\Vert \one_{\CC_{a,b}} ( - \Delta_{D} - \lambda^{2} )^{- 1} \one_{\CC_{a,b}} \big\Vert ,
\end{equation}
for $\lambda$ not resonance with $\re \lambda \geq 1$ and $0 \geq \im \lambda \geq - M \ln ( \re \lambda )$.
\end{theorem}

For $n \geq 3$, $n$ odd, there is a link between the cut-off resolvent $\varphi (-\Delta_D - \lambda^2)^{-1} \varphi$ and
the contraction semigroup $Z^{\rho} (t) = P_{+}^{\rho} e^{-i t G} P_{-}^{\rho} = e^{t B^{\rho}} : H \longrightarrow H$, $t \geq 0$, with generator $B^{\rho}$, introduced by Lax and Phillips \cite{LaPh89_01}. Here, $P_{\pm}^{\rho}$ are the orthogonal projections on the orthogonal complements of the Lax--Phillips spaces $D_{\pm}^{\rho}$, ${\rho} > R_{0}$.
The spectrum of $i B^{\rho}$ coincides with the resonances and is then independent on the choice of $\rho > R_{0}$. Given $\varphi \in C_0^{\infty} ( \Omega )$, we may fix $\rho > R_{0}$ so that $\varphi P_{\pm}^{\rho} = \varphi = P_{\pm}^{\rho} \varphi$. In the sequel, we drop the indexes $\rho$ in the notations and write $B , P_{\pm}$ instead of $B^{\rho} , P_{\pm}^{\rho}$. For $\im \lambda > 0$, we have
\begin{equation*}
- \varphi ( B + i \lambda )^{- 1} \varphi = \int_0^{\infty} e^{i \lambda t} \varphi P_{+} e^{- i t G} P_{-} \varphi d t = - i \varphi ( G - \lambda )^{-1} \varphi ,
\end{equation*}
and, by analytic continuation, this equality holds true for $\lambda$ not resonance with $\im \lambda \leq 0$. Moreover, one can see that
\begin{equation*}
\big\Vert \varphi ( G - \lambda )^{- 1} \varphi \big\Vert_{H \to H} \leq C \big\Vert \varphi \lambda ( - \Delta_D - \lambda^2 )^{- 1} \varphi \big\Vert_{L^2 ( \Omega ) \to L^2 ( \Omega )} ,
\end{equation*}
for $\lambda$ not resonance with $\vert \lambda \vert \geq 1$. Thus \eqref{a44} implies
\begin{equation} \label{a46}
\big\Vert \varphi ( B + i \lambda )^{- 1} \varphi \big\Vert_{H \to H} \leq C \frac{\vert \lambda \vert}{\vert \im \lambda \vert} e^{C\vert \im \lambda \vert} \big\Vert \one_{\CC_{a,b}} ( - \Delta_D - \lambda^2 )^{- 1} \one_{\CC_{a,b}} \big\Vert_{L^2 ( \Omega ) \to L^2 ( \Omega )} .
\end{equation}

Note that, in odd dimension $n \geq 3$, it is possible to estimate the cut-off resolvent in term of scattering quantities. This was done by Stoyanov and the second author in \cite{PeSt09_01} using the Lax--Phillips theory. More precisely, consider the scattering matrix $S ( \lambda ) = I + K ( \lambda ): L^2 ( \S^{n - 1} ) \longrightarrow L^2 ( \S^{n - 1} )$, associated to the Dirichlet problem for the wave equation in $\Omega$ given in \eqref{a43}. This operator is defined for $\im \lambda \geq 0$ and it is unitary for $\lambda \in \R$. The operator $K ( \lambda )$ is a Hilbert--Schmidt operator with kernel $a ( \lambda , \omega , \omega^{\prime} )$, called scattering amplitude. The scattering matrix $S ( \lambda )$ (as the scattering amplitude $a ( \lambda , \omega , \omega^{\prime} )$) has a meromorphic continuation from $\im \lambda \geq 0$ to the half plane $\im \lambda < 0$ and the poles coincide with the resonances. Of course, the form of the scattering operator $S ( \lambda )$ depends on the outgoing and incoming representations of the energy space $H$ (see \cite{LaPh89_01}) and to have the formula \eqref{a52} for the scattering amplitude we must have an appropriate outgoing/incoming representation.

By using the link between $\Vert ( B + i \lambda )^{- 1} \Vert_{H \to H}$ and the inner representation of the scattering operator $S_1(\lambda)$ established in \cite[Chapter IV]{LaPh89_01}, it is proved in \cite[Section 4]{PeSt09_01} that
\begin{equation}
\Vert ( B + i \lambda )^{- 1} \Vert_{H \to H} \leq \frac{3}{2} \frac{e^{\beta \vert \im \lambda \vert}}{\vert \im \lambda \vert} \Vert S ( \lambda ) \Vert_{L^2 ( \S^{n - 1} ) \to L^2 ( \S^{n - 1} )} ,
\end{equation}
for some $\beta \geq 0$ given by the inner representation of the scattering operator. Using that the Hilbert--Schmidt norm of an operator is the $L^{2}$ norm of its kernel, the last estimate yields
\begin{equation} \label{a47}
\Vert ( B + i \lambda )^{- 1} \Vert_{H \to H} \leq \frac{3}{2} \frac{e^{\beta \vert \im \lambda \vert}}{\vert \im \lambda \vert} \bigg( \bigg( \int_{\S^{n-1} \times \S^{n-1}} \big\vert a ( \lambda , \omega^{\prime} , \omega ) \big\vert^{2} d \omega \, d \omega^{\prime} \bigg)^{1/2} + 1 \bigg) .
\end{equation}
Now, we can handle the integral over $\S^{n-1} \times \S^{n-1}$ using the representation \eqref{a52} with $h = 1$, $z = \lambda^2$ and $P = - \Delta_D$. Choosing the functions $\chi_{j} \in C_0^{\infty} ( \Omega )$, $j = 1 , 2$ so that $\nabla \chi_{j} \prec \one_{\CC_{a,b}}$, the formula \eqref{a52} and the estimate \eqref{a47} give an analog of \eqref{a46} with a possible polynomial loss in $\< \lambda \>$.

\appendix

\section{Properties of the generalized resonant states} \label{s7}

In this part, we collect some basic properties of the generalized resonant states. Being for the most part in the folklore of the theory of resonances, we only give them for a reason of completeness.

Let $z \in \Lambda_{\theta_{0}}$ be a resonance of $P$. Since $( P - \lambda )^{-1} : \CH_{\rm comp} \longrightarrow \CD_{\rm loc}$ is an operator-valued meromorphic function, we can write, for $\lambda$ in a neighborhood of $z$,
\begin{equation*}
( P - \lambda )^{-1} = \frac{\Pi_{N}}{(z - \lambda )^{N}} + \cdots + \frac{\Pi_{1}}{z - \lambda} + \CA ( \lambda ) ,
\end{equation*}
as operators from $\CH_{\rm comp}$ to $\CD_{\rm loc}$, where $\CA ( \lambda )$ is holomorphic near $z$ and the $\Pi_{j}$'s are finite rank operators. Let $P_{\theta}$ be a complex distortion by an angle $\arctan \big( \frac{\vert \im z \vert}{\vert \re z \vert} \big) < \theta \leq \theta_{0}$. Then, for $\lambda$ in a neighborhood of $z$, we have
\begin{equation*}
( P_{\theta} - \lambda )^{-1} = \frac{\Pi^{\theta}_{N_{\theta}}}{(z - \lambda )^{N_{\theta}}} + \cdots + \frac{\Pi^{\theta}_{1}}{z - \lambda} + \CA ( \lambda ) ,
\end{equation*}
as operators from $\CH$ to $\CD$, where $\CA ( \lambda )$ is holomorphic near $z$ and the $\Pi_{j}^{\theta}$'s are finite rank operators. Moreover, if the distortion occurs outside of the support of $\varphi \in C^{\infty}_{0}( \R^{n} )$, it follows from \eqref{a34} that
\begin{equation} \label{a35}
\varphi \Pi_{j} \varphi = \varphi \Pi_{j}^{\theta} \varphi ,
\end{equation}
for all $j \geq 1$.

\begin{lemma}\sl \label{a33}
Let $\varphi \in C^{\infty}_{0} ( \R^{n} )$ be such that $\one_{B ( R_{1} )} \prec \varphi$. Then, the multiplication by $\varphi$ is injective on $\im \Pi_{j}$ (resp. $\im \Pi^{\theta}_{j}$) for all $1 \leq j \leq N$ (resp. $1 \leq j \leq N_{\theta}$).
\end{lemma}

\begin{proof}
Let $u_{\theta} \in \im \Pi_{j}^{\theta}$ be such that $\varphi u_{\theta} = 0$. Using $( P_{\theta} - z ) \Pi^{\theta}_{N_{\theta}} = 0$ and $( P_{\theta} - z ) \Pi^{\theta}_{k} = \Pi^{\theta}_{k + 1}$, we get
\begin{equation*}
( P_{\theta} - z ) ( P_{\theta} - z )^{N_{\theta} - 1} u_{\theta} = ( P_{\theta} - z )^{N_{\theta}} u_{\theta} = 0 .
\end{equation*}
From Lemma 3.1 of Sj\"{o}strand and Zworski \cite{SjZw91_01}, we deduce that $( P_{\theta} - z )^{N_{\theta} - 1} u_{\theta}$ is (outside of $B ( R_{1} )$) the restriction to $\Gamma_{\theta}$ of a holomorphic function in $\Upsilon$. On the other hand, $( P_{\theta} - z )^{N_{\theta} - 1} u_{\theta} = 0$ on the support of $\varphi$ since $\varphi u_{\theta} = 0$. Therefore,
\begin{equation*}
( P_{\theta} - z ) ( P_{\theta} - z )^{N_{\theta} - 2} u_{\theta} = ( P_{\theta} - z )^{N_{\theta} - 1} u_{\theta} = 0 .
\end{equation*}
Then, performing an induction argument, we get $u_{\theta} = 0$. The fact that the multiplication by $\varphi$ is injective on $\im \Pi_{j}$ is similar.
\end{proof}

\begin{remark}\sl \label{a36}
Using ${( P - \lambda )^{- 1}}^{*} = ( P - \bar{\lambda} )^{-1}$ (resp. ${( P_{\theta} - \lambda )^{- 1}}^{*} = ( P_{- \theta} - \bar{\lambda} )^{-1}$), we can prove the same way that $\im \Pi_{j} \varphi = \im \Pi_{j}$ (resp. $\im \Pi^{\theta}_{j} \varphi = \im \Pi^{\theta}_{j}$).
\end{remark}

Combining \eqref{a35}, Lemma \ref{a33} and Remark \ref{a36}, we get

\begin{corollary}\sl \label{a37}
Let $\varphi \in C^{\infty}_{0} ( \R^{n} )$ be such that $\one_{B ( R_{1} )} \prec \varphi$ and such that the distortion occurs outside of the support of $\varphi$. Then, we have $N = N_{\theta}$ and, for all $1 \leq j \leq N$,
\begin{equation*}
\rank \Pi_{j} = \rank \varphi \Pi_{j} \varphi = \rank \varphi \Pi_{j}^{\theta} \varphi = \rank \Pi_{j}^{\theta} .
\end{equation*}
In particular,
\begin{equation} \label{a40}
\im \varphi \Pi_{j} = \im \varphi \Pi_{j} \varphi = \im \varphi \Pi_{j}^{\theta} \varphi = \im \varphi \Pi_{j}^{\theta} .
\end{equation}
\end{corollary}

\begin{lemma}\sl \label{a38}
For all $1 \leq j \leq N$, we have $\im \Pi_{j} \subset \im \Pi_{1}$ and $\im \Pi_{j}^{\theta} \subset \im \Pi_{1}^{\theta}$.
\end{lemma}

\begin{proof}
Since the resolvent of $P_{\theta}$ acts from $L^{2} ( \R^{n} )$ to itself, a standard argument gives $\im \Pi_{j}^{\theta} \subset \im \Pi_{1}^{\theta}$. Consider now $u \in \im \Pi_{j}$. Let $\varphi \in C^{\infty}_{0} ( \R^{n} )$ be such that $\one_{B ( R_{1} )} \prec \varphi$ and $P_{\vartheta}$ be a complex distortion outside the support of $\varphi$. Then, from \eqref{a40}, there exists $u_{\vartheta} \in \im \Pi_{j}^{\vartheta}$ such that $\varphi u = \varphi u_{\vartheta}$. Therefore, using $\im \Pi_{j}^{\vartheta} \subset \im \Pi_{1}^{\vartheta}$ together with \eqref{a40}, there exists $u_{\varphi} \in \im \Pi_{1}$ such that
\begin{equation*}
\varphi u = \varphi u_{\varphi} .
\end{equation*}
Let now $\psi \in C^{\infty}_{0} ( \R^{n} )$ be such that $\varphi \prec \psi$. From the previous construction, $\varphi u_{\psi} = \varphi \psi u_{\psi} = \varphi \psi u = \varphi u = \varphi u_{\varphi}$ and $ u_{\varphi} - u_{\psi} \in \im \Pi_{1}$. Then, Lemma \ref{a33} implies $u_{\varphi} = u_{\psi}$. In other words, for all $\psi \in C^{\infty}_{0}( \R^{n} )$, we have
\begin{equation*}
\psi u = \psi u_{\varphi} .
\end{equation*}
This implies $u = u_{\varphi} \in \im \Pi_{1}$.
\end{proof}

\begin{lemma}\sl \label{a39}
Let $\varphi \in C^{\infty}_{0} ( \R^{n} )$ be such that $\one_{B ( R_{1} )} \prec \varphi$ and such that the distortion occurs outside of the support of $\varphi$. Then, for all $u \in \im \Pi_{1}$, there exists a unique $u_{\theta} \in \im \Pi_{1}^{\theta}$ such that $\varphi u = \varphi u_{\theta}$. Moreover, $( P - z )^{J} u = 0$ if and only if $( P_{\theta} - z )^{J} u_{\theta} = 0$.
\end{lemma}

\begin{proof}
Let $u \in \im \Pi_{1}$. From \eqref{a40}, there exists $u_{\theta} \in \im \Pi_{j}^{\theta}$ such that $\varphi u = \varphi u_{\theta}$. Thanks to Lemma \ref{a33}, this $u_{\theta}$ is unique. Lemma \ref{a38}, $( P - z ) \Pi_{j} = \Pi_{j + 1}$ and $( P_{\theta} - z ) \Pi_{j}^{\theta} = \Pi_{j + 1}^{\theta}$ imply $( P - z )^{J} u \in \im \Pi_{1}$ and $( P_{\theta} - z )^{J} u_{\theta} \in \im \Pi_{1}^{\theta}$. Then, from Lemma \ref{a33}, $( P - z )^{J} u = 0$ iff $\varphi ( P - z )^{J} u = \varphi ( P_{\theta} - z )^{J} u_{\theta} = 0$ iff $( P_{\theta} - z )^{J} u_{\theta} = 0$.
\end{proof}

{\bf Acknowledgments.} The authors would like to thank the referee for helpful comments, making the paper more understandable.

\bibliographystyle{amsplain}

\begin{thebibliography}{10}

\bibitem{AgCo71_01}
J.~Aguilar and J.~M. Combes, \emph{A class of analytic perturbations for
  one-body {S}chr\"odinger {H}amiltonians}, Comm. Math. Phys. \textbf{22}
  (1971), 269--279.

\bibitem{BoFuRaZe09_01}
J.-F. Bony, S.~Fujii\'e, T.~Ramond, and M.~Zerzeri, \emph{Spectral projection,
  residue of the scattering amplitude, and Schr\"odinger group expansion for
  barrier-top resonances}, Ann. Inst. Fourier \textbf{61} (2011), no.~4, 1351--1406.  
  
\bibitem{BoMi04_01}
J.-F. Bony and L.~Michel, \emph{Microlocalization of resonant states and
  estimates of the residue of the scattering amplitude}, Comm. Math. Phys.
  \textbf{246} (2004), no.~2, 375--402.

\bibitem{BrPe00_01}
V.~Bruneau and V.~Petkov, \emph{Semiclassical resolvent estimates for trapping
  perturbations}, Comm. Math. Phys. \textbf{213} (2000), no.~2, 413--432.

\bibitem{Bu98_01}
N.~Burq, \emph{D\'ecroissance de l'\'energie locale de l'\'equation des ondes
  pour le probl\`eme ext\'erieur et absence de r\'esonance au voisinage du
  r\'eel}, Acta Math. \textbf{180} (1998), no.~1, 1--29.

\bibitem{Bu02_01}
N.~Burq, \emph{Lower bounds for shape resonances widths of long range
  {S}chr\"odinger operators}, Amer. J. Math. \textbf{124} (2002), no.~4,
  677--735.

\bibitem{BuZw01_01}
N.~Burq and M.~Zworski, \emph{Resonance expansions in semi-classical
  propagation}, Comm. Math. Phys. \textbf{223} (2001), no.~1, 1--12.

\bibitem{CaVo02_01}
F.~Cardoso and G.~Vodev, \emph{Uniform estimates of the resolvent of the
  {L}aplace-{B}eltrami operator on infinite volume {R}iemannian manifolds.
  {II}}, Ann. Henri Poincar\'e \textbf{3} (2002), no.~4, 673--691.

\bibitem{DiSj99_01}
M.~Dimassi and J.~Sj{\"o}strand, \emph{Spectral asymptotics in the
  semi-classical limit}, London Mathematical Society Lecture Note Series, vol.
  268, Cambridge University Press, Cambridge, 1999.

\bibitem{FuLaMa11_01}
S.~Fujii{\'e}, A.~Lahmar-Benbernou, and A.~Martinez, \emph{Width of shape
  resonances for non globally analytic potentials}, J. Math. Soc. Japan
  \textbf{63} (2011), no.~1, 1--78.

\bibitem{HaSi80_01}
E.~Harrell and B.~Simon, \emph{The mathematical theory of resonances whose
  widths are exponentially small}, Duke Math. J. \textbf{47} (1980), no.~4,
  845--902.

\bibitem{HeMa87_01}
B.~Helffer and A.~Martinez, \emph{Comparaison entre les diverses notions de
  r\'esonances}, Helv. Phys. Acta \textbf{60} (1987), no.~8, 992--1003.

\bibitem{HeSj86_01}
B.~Helffer and J.~Sj{\"o}strand, \emph{R\'esonances en limite semi-classique},
  M\'em. Soc. Math. France (1986), no.~24-25, iv+228.

\bibitem{Hu86_01}
W.~Hunziker, \emph{Distortion analyticity and molecular resonance curves}, Ann.
  Inst. H. Poincar\'e Phys. Th\'eor. \textbf{45} (1986), no.~4, 339--358.

\bibitem{LaPh89_01}
P.~Lax and R.~Phillips, \emph{Scattering theory}, second ed., Pure and Applied
  Mathematics, vol.~26, Academic Press Inc., 1989, With appendices by C.
  Morawetz and G. Schmidt.

\bibitem{Ma02_01}
A.~Martinez, \emph{Resonance free domains for non globally analytic
  potentials}, Ann. Henri Poincar\'e \textbf{3} (2002), no.~4, 739--756.

\bibitem{NaStZw03_01}
S.~Nakamura, P.~Stefanov, and M.~Zworski, \emph{Resonance expansions of
  propagators in the presence of potential barriers}, J. Funct. Anal.
  \textbf{205} (2003), no.~1, 180--205.

\bibitem{NoZw09_01}
S.~Nonnenmacher and M.~Zworski, \emph{Quantum decay rates in chaotic
  scattering}, Acta Math. \textbf{203} (2009), no.~2, 149--233.

\bibitem{PeSt09_01}
V.~Petkov and L.~Stoyanov, \emph{Singularities of the scattering kernel related
  to trapping rays}, Advances in phase space analysis of partial differential
  equations, Progr. Nonlinear Differential Equations Appl., vol.~78,
  Birkh\"auser Boston Inc., 2009, pp.~235--251.

\bibitem{PeZw01_01}
V.~Petkov and M.~Zworski, \emph{Semi-classical estimates on the scattering
  determinant}, Ann. Henri Poincar\'e \textbf{2} (2001), no.~4, 675--711.

\bibitem{RoTa87_01}
D.~Robert and H.~Tamura, \emph{Semiclassical estimates for resolvents and
  asymptotics for total scattering cross-sections}, Ann. Inst. H. Poincar\'e
  Phys. Th\'eor. \textbf{46} (1987), no.~4, 415--442.

\bibitem{Sj97_01}
J.~Sj{\"o}strand, \emph{A trace formula and review of some estimates for
  resonances}, Microlocal analysis and spectral theory (Lucca, 1996), NATO Adv.
  Sci. Inst. Ser. C Math. Phys. Sci., vol. 490, Kluwer Acad. Publ., 1997,
  pp.~377--437.

\bibitem{Sj01_01}
J.~Sj{\"o}strand, \emph{Resonances for bottles and trace formulae}, Math. Nachr.
  \textbf{221} (2001), 95--149.

\bibitem{Sj07_01}
J.~Sj{\"o}strand, \emph{Lectures on resonances}, preprint on {\it
  http://www.math.polytechnique.fr/$\sim$sjoestrand} (2007), 1--169.

\bibitem{SjZw91_01}
J.~Sj{\"o}strand and M.~Zworski, \emph{Complex scaling and the distribution of
  scattering poles}, J. Amer. Math. Soc. \textbf{4} (1991), no.~4, 729--769.

\bibitem{St02_01}
P.~Stefanov, \emph{Estimates on the residue of the scattering amplitude},
  Asymptot. Anal. \textbf{32} (2002), no.~3-4, 317--333.

\bibitem{TaZw98_01}
S.-H. Tang and M.~Zworski, \emph{From quasimodes to reasonances}, Math. Res.
  Lett. \textbf{5} (1998), no.~3, 261--272.

\bibitem{VaZw00_01}
A.~Vasy and M.~Zworski, \emph{Semiclassical estimates in asymptotically
  {E}uclidean scattering}, Comm. Math. Phys. \textbf{212} (2000), no.~1,
  205--217.
\end{thebibliography}

\end{document}